\newif\ifsubmit

\submittrue

\ifsubmit %
\documentclass[11pt]{article} %
\else %
\documentclass[oneside]{article} %
\fi

\usepackage{article} %
\usepackage{math} %
\usepackage{algo} %
\usepackage{wrapfig} %
\usepackage{placeins} %
\usepackage{graphicx} %
\usepackage[longnamesfirst]{natbib} %
\usepackage{setspace} %
\usepackage{christofides}
\usepackage{multirow}           %

\usepackage[font={small},labelfont=bf
]{caption}

\ifsubmit \usepackage[letterpaper]{geometry}
\newcommand{\metrictsp}{\text{Metric-TSP}\xspace}

\else %

\fi

\everymath{\displaystyle}

\begin{document}

\title{Faster Approximations for \metrictsp via \mbox{Linear
    Programming}\thanks{Department of Computer Science, University of Illinois,
Urbana-Champaign, Urbana, IL 61801. {\tt \{chekuri,quanrud2\}@illinois.edu}.
Work on this paper is partly supported by NSF grant CCF-1526799.}
} %

\author{
Chandra Chekuri
\and
Kent Quanrud
}
\maketitle

\setcounter{tocdepth}{3}

\onehalfspacing

\begin{abstract}
  We develop faster approximation algorithms for \metrictsp building
  on recent, nearly linear time approximation schemes for the LP
  relaxation \citep{cq-17-focs}. We show that the LP solution can be
  sparsified via cut-sparsification techniques such as those of
  \citet{bk-15}.  Given a weighted graph $G$ with $m$ edges and $n$
  vertices, and $\eps > 0$, our randomized algorithm outputs with high
  probability a $(1+\eps)$-approximate solution to the LP relaxation
  whose support has $O(n \log n/\eps^2)$ edges. The running time of
  the algorithm is $\apxO{m/\eps^2}$. This can be generically used to
  speed up algorithms that rely on the LP.

  For \metrictsp, we obtain the following concrete result. For a
  weighted graph $G$ with $m$ edges and $n$ vertices, and $\eps > 0$,
  we describe an algorithm that outputs with high probability a tour
  of $\graph$ with cost at most
  \begin{math}
    \epsmore \frac{3}{2}
  \end{math}
  times the minimum cost tour of $\graph$ in time
  \begin{math}
    \apxO{\frac{m}{\eps^2} + \frac{n^{1.5}}{\eps^3}}.
  \end{math}
  Previous implementations of Christofides' algorithm
  \citep{christofides-76} require, for a $\frac{3}{2}$-optimal tour,
  $\apxO{n^{2.5}}$ time when the metric is explicitly given, or
  $\apxO{\min{m^{1.5}, mn+n^{2.5}}}$ time when the metric is given
  implicitly as the shortest path metric of a weighted graph.
\end{abstract}

\section{Introduction}
\label{sec:intro}
The traveling salesman problem (abbr.\ TSP) and its variants including
\metrictsp are extensively studied in discrete and combinatorial
optimization. In this short paper we focus on approximation algorithms
for \metrictsp which is NP-Hard.  An instance of \metrictsp is
specified by a complete undirected graph $G=(V,E,c)$ with positive
edges costs $c: E \rightarrow \preals$ satisfying the triangle
inequality (and form a metric over $V$). The goal is to find a
Hamiltonian cycle in $G$ of minimum cost.  Note that if $G$ is not
required to be a complete graph, then the minimum cost Hamiltonian
cycle is inapproximable, as deciding if an undirected graph has a
Hamiltonian cycle is NP-Complete.

\begin{definition}
  An instance of \metrictsp is \emph{explicit} if the input is a
  metric on $n$ nodes with all $n \choose 2$ distances specified.  An
  instance of \metrictsp is \emph{implicit} if it is specified
  implicitly as the metric completion of an underlying graph
  $\defweightedgraph$.
\end{definition}

\begin{remark}
  A sparse graph with $m$ edges and $n$ vertices generates a metric
  TSP problem that is inherently of size $\bigO{n^2}$. Given an
  implicit instance of \metrictsp, it is desirable to obtain running
  times relative to the number of edges $m$ in the underlying graph,
  which represents the true input size of the problem.
\end{remark}

Many instances of \metrictsp are
implicitly defined, either as a sparse graph, or in other settings
where there is an easy function that returns $c(u,v)$ given nodes
$u,v$ (geometric instances in low dimensions are a good example).
In this paper we focus on the setting where $G$ is specified as a
sparse weighted graph. In such a setting, finding a solution corresponds
to finding a tour of the vertices of minimum total cost. This is equivalent
to finding a minimum cost Eulerian \emph{multigraph} in the support of the
given graph $G=(V,E,c)$.

\cite{christofides-76} described a $\frac{3}{2}$-approximation for \metrictsp
and this is still the best known. There is a well-known conjecture
that a $\frac{4}{3}$-approximation is achievable via a solution to an LP
relaxation called the \emph{subtour elimination LP} due to
\cite{dfj-54}; see \refsection{lp} for a description of this LP
referred to as \refequation{subtour-elimination} and also a related LP
\refequation{2ecss} whose optimum values are equal for instances of
\metrictsp. There has been recent exciting progress on various special
cases for which we refer the reader to a survey by \cite{vygen-12}.

The goal of this paper is to find \emph{fast} approximation algorithms
for \metrictsp. One can easily obtain a $2$-approximation in nearly
linear time by traversing the minimum spanning tree (abbr.\
MST). Christofides's algorithm requires the computation of a
minimum-cost perfect matching (see \refsection{christofides} for more
details). The fastest known implementation requires $\apxO{n^{2.5}}$
time when the metric is explicitly given and
$\apxO{\min(m^{1.5}, mn+n^{2.5})}$ time when it is given implicitly as
a weighted graph \citep{gt-91}. In a recent paper we obtained the
following result to compute a nearly linear time approximation scheme
for the LP relaxation \refequation{2ecss}.

\begin{theorem}[{\citealp{cq-17-focs}}]
  \labeltheorem{apx-2ecss} There is a randomized algorithm that in
  $\apxO{m/\eps^2}$ time, with high probability, computes a feasible point $x \in
  \reals^{\edges}$ for \refequation{2ecss} with objective value
  $\sum_{e \in \edges} \cost{e} x_e \leq \epsmore
  \opt\refequation{2ecss}$.
\end{theorem}

Surprisingly the running time to solve the LP is significantly faster than
the time to implement Christofides's heuristic. \citet{cq-17-focs}
raised the question of faster algorithms that yield
a $\frac{3}{2}$-approximation. In this paper we make progress towards the question.
Our first result shows that the LP solution can be sparsified as follows.

\begin{theorem}
  \labeltheorem{sparsifying-2ecss} There is a randomized algorithm that
  given a feasible solution $x \in \reals^{\edges}$ for
  \refequation{2ecss} and $\eps > 0$, outputs, with high probability,
  another feasible solution $x'$ such that (i) support of $x'$ is
  $O(n \log n/\eps^2)$ and (ii) the objective value of $x'$ is
  close to that of $x$, that is, $\sum_{e \in \edges}
  \cost{e} x'_e \leq \epsmore \sum_{e \in \edges} \cost{e} x_e$.
\end{theorem}

Combining the two preceding theorems yields a randomized algorithm
that, in nearly linear time, sparsifies a given graph $G$ on $n$ nodes
and $m$ edges to a subgraph $H$ with $O(n \log n/\eps^2)$ edges such
that the LP value on $H$ is within a $(1+\eps)$-factor of the LP value
on $G$. This sparsification can generically help any approximation
algorithm that relies on the LP solution either directly or
indirectly.  We use cut sparsification techniques for the above. The
only novelty is that, unlike most applications that we are aware of, we
also need to preserve the \emph{cost} of the sparsifier; we observe
that importance sampling which underlies a class of sparsifiers
\citep{bk-15,fhhp-11} is suitable for this purpose. Stronger
sparsification results such as the one of \cite{bss-12} do not appear
to preserve the cost.

We utilize the two preceding theorems and the analysis by
\cite{wolsey-80} of Christofides' heuristic with respect to the LP, to
obtain a fast $\parof{\frac{3}{2} + \eps}$-approximation.

\begin{theorem}
  \labeltheorem{fast-christofides-intro}
  There is a randomized algorithm that in
  \begin{math}
    \apxO{\frac{m}{\eps^2} + \frac{n^{1.5}}{\eps^3}}
  \end{math}
  time, with probability at least $1 - \frac{1}{\poly{m}}$,
  returns a tour of
  cost at most
  \begin{math}
    \epsmore \frac{3}{2} \opt \refequation{subtour-elimination}.
  \end{math}
\end{theorem}

\section{Preliminaries}

\begin{table}
  \centering \renewcommand{\arraystretch}{3} %
  \normalsize
  \begin{tabular}{| c | c | c |}
    \hline                      %
    Metric        %
    & $\frac{3}{2}$-APX         %
    & $\epsmore\frac{3}{2}$-APX
    \\
    \hline                           %
    \multirow{2}{4em}{\centering Shortest paths} %
    &                                %
      $\apxO{m n + n^{2.5}}$ & \multirow{2}{*}{$\apxO{\frac{m}{\eps^2} +
                               \frac{n^{1.5}}{\eps^3}}$} \\
    \cline{2-2}                 %
    & $\apxO{m^2}$ & \\
    \hline
    Explicit
    &
      \begin{math}
        \apxO{n^{2.5}}
      \end{math}
    &
      \begin{math}
        \apxO{\frac{n^2}{\eps^2} +
          \frac{n^{1.5}}{\eps^3}}
      \end{math}
    \\
    \hline
  \end{tabular}

  \medskip

  \begin{minipage}{206pt}
      \caption{Running times for \metrictsp.  The running times for
        computing an $\epsmore \frac{3}{2}$ are new and the algorithms
        are randomized and work with high probability.  }
    \end{minipage}
\end{table}

\subsection{Christofides' $\frac{3}{2}$-approximation algorithm}

\labelsection{christofides}

In this section, we review the approximation algorithm of
\citet{christofides-76} for \metrictsp.

\begin{definition}
  Given an even set of vertices $T \subseteq \vertices$, a
  \emph{$T$-join} is a subgraph $H$ of $\graph$ for which $T$ is the
  set of vertices with odd degree. The cost of a $T$-join is the sum
  cost of all the edges in $H$,
  \begin{math}
    \sumcost{H} = \sum_{e \in \edges{H}} \cost{e}.
  \end{math}
\end{definition}
$T$-joins and polynomial-time algorithms for computing minimum
$T$-joins are discussed later in \refsection{t-joins}.

\begin{center}
  \begin{minipage}{.5\paperwidth}
    \ttfamily
    \noindent{\underline{\citet{christofides-76}}}
    \begin{enumerate}
    \item Compute the minimum spanning tree $M$ of $\graph$.
    \item Compute the minimum cost $T$-join $H$ of $\graph$, where $T$ is
      the set of odd-degree vertices in $M$.
    \item The multiset $M + H$ is an Eulerian graph. Return an Eulerian
      tour of $M + H$.
    \end{enumerate}
  \end{minipage}
\end{center}

\begin{theorem}[{\citealp{christofides-76}}]
  Christofides' algorithm returns a tour of cost at most $\frac{3}{2}$
  times the optimal value.
\end{theorem}
\begin{proof}[Proof sketch]
  As the minimum cost connected subgraph of $\graph$ the minimum
  spanning tree has cost at most $\parof{1 - 1/n}\opt$. The minimum cost $T$-join has
  cost at most $1/2$ times $\opt$, which can be seen as follows. Break
  the optimal tour into paths between vertices in $T$. There is an
  even number of these paths, and taking every other path induces a
  $T$-join. That is, the optimal tour can be divided into 2
  $T$-joins. The smaller of the two $T$-joins has cost at most half of
  the optimal tour.
\end{proof}

\subsection{LP Relaxations for Metric TSP}

\labelsection{lp}

A standard LP for TSP is the following \emph{subtour elimination LP.}
\begin{align*}
  \begin{aligned}
    \text{minimize }                    %
    & \sum_{e \in \edges} \cost{e} y_e %
    \text{ over } y \in \reals^{\edges}               %
    \\
    \text{ s.t.\ }                 %
    & \sum_{e \in \delta(v)} y_e = 2 \text{ for all } v \in \vertices,\\
    & \sum_{e \in \delta(U)} y_e    %
    \geq                          %
    2 \text{ for all } \emptyset
    \subsetneq U \subsetneq \vertices, \\
    \text{and }                   %
    &                 %
    y_e \in [0,1] \text{ for all } e \in \edges.
  \end{aligned}
      \labelthisequation[\textsf{SE}]{subtour-elimination}
\end{align*}
Here $\delta(S)$ denotes the set of edges crossing a set $S \subset
V$.  The first set of constraints require each vertex to be incident
to exactly two edges (in the integral setting), and are referred to as
degree constraints. The second constraint forces two-edge connectivity.
The LP provides a lower bound for TSP that coincides with the
lower bound of \citet{hk-70}. The worst-case integrality gap of
this LP is conjectured to be $\frac{4}{3}$ and is a major open problem in
approximation algorithms. \cite{wolsey-80} showed that the integerality gap
is at most $\frac{3}{2}$ by analyzing Christofides's algorithm via
the LP.

\begin{fact}[\citet{wolsey-80}]
  Christofides' algorithm returns a tour of cost $\leq \frac{3}{2}$
  times the cost of $\opt\refequation{subtour-elimination}$. In
  particular, the integrality gap of \refequation{subtour-elimination}
  is at most $\frac{3}{2}$.
\end{fact}
To apply the lower bound to an
implicit instance of \metrictsp defined by $\graph$, one needs to
apply it to the metric completion of $\graph$. Instead one can consider
a simpler LP obtained by dropping the degree constraints and the
upper bound constraints in \refequation{subtour-elimination}.
This leads to the following LP relaxation for the \emph{2-edge connected
  spanning subgraph problem (allowing multiplicities)}.
\begin{align*}
  \begin{aligned}
    \text{minimize } %
    & \sum_{e \in \edges} \cost{e} y_e %
    \text{ over } y \in \reals^{\edges} %
    \\
    \text{ s.t.\ } %
    & \sum_{e \in \delta(U)} y_e %
    \geq %
    2 \text{ for all } \emptyset
    \subsetneq U \subsetneq \vertices \\
    \text{and } %
    & %
    y_e \geq 0 \text{ for all } e \in \edges.
  \end{aligned}
      \labelthisequation[\textsf{2ECSS}]{2ecss}
\end{align*}
\begin{fact}[{Cunningham [via \citealp{mmp-90}], \citealp{gb-93}}]
  \labelfact{subtour-elimination=2ecss} %
  The optimum value of subtour elimination LP
  \refequation{subtour-elimination} for the metric completion of
  $\graph$ coincides with the optimum value of the 2-edge connected
  spanning subgraph LP \refequation{2ecss}.
\end{fact}

\subsection{Perfect matchings and $T$-joins}
\labelsection{t-joins}

Before addressing $T$-joins in general, we first review the following
special case where $T = V$.
\begin{definition}
  A \emph{matching} in $\graph$ is a set of edges $M \subseteq E$ such
  that each vertex is incident to at most one edge in $M$. A
  \emph{perfect matching} is a set of edges $M \subseteq E$ such that
  each vertex is incident to exactly one edge in $M$. The cost of a
  matching $M$ is $\sum_{e \in E} \cost{e}$.
\end{definition}

\begin{problem}[Minimum cost perfect matching]
  Assuming $\graph$ contains a perfect matching, compute the minimum
  cost perfect matching.
\end{problem}

The following running times are known for computing minimum cost
perfect matchings.

\begin{fact}[{\citealp{gt-91}}]
  \labelfact{gt-91} If the edge costs are integers between $-W$ and
  $W$, the minimum cost perfect matching can be computed in
  $\apxO{m \sqrt{n} \log W}$ time. For general edge weights and any
  $\eps > 0$, an $\epsmore$-minimum perfect matching can be computed
  in $\apxO{m \sqrt{n} \log \reps}$ time.
\end{fact}



Clearly, computing the minimum weight perfect matching reduces to
computing the minimum cost $T$-join for $T = V$. Conversely, we have
the following.

\begin{fact}[{\citealp{edmonds-65}}]
  \labelfact{edmonds-65}
  If all edge weights are non-negative, then the minimum cost $T$-join
  equals the minimum weight perfect matching on the clique with vertex
  set $T$, where the weight of an edge $(s,t) \in T \times T$ is the
  length of the shortest path from $s$ to $t$ in $\graph$.
\end{fact}

\begin{fact}[{\citealp{gt-91}}]
  A minimum $T$-join can be computed in
  time; an $\epsmore$-minimum $T$-join can be computed in
  \begin{math}
    \apxO{\sizeof{T} m + \sizeof{T}^{2.5} \log \reps}.
  \end{math}
\end{fact}

\begin{proof}[Proof sketch]
  For non-negative edge costs, all shortest paths between vertices in
  $T$ can be found in $\apxO{\sizeof{T} m }$ time with Dijkstra's
  shortest path algorithms. Submitting these lengths to the matching
  algorithms of \citet{gt-91} (\reffact{gt-91}) gives the desired
  running times.
\end{proof}
A different reduction from $T$-joins to perfect matching, better
suited for sparse graphs and $\sizeof{T}$ large, is the following.
\begin{fact}[{\citealp[Theorem 3]{bkvz-99}}]
  \labelfact{t-join-gadget} The minimum cost $T$-join is equivalent to
  the minimum cost perfect matching of an auxiliary graph with
  $\bigO{m}$ nodes and $\bigO{m}$ edges. The auxiliary graph can be
  computed in $\bigO{m}$ time.
\end{fact}
\begin{fact}
  \labelfact{sparse-t-join} The minimum $T$-join can be computed in
  $\apxO{m^{1.5} \log W}$ time if the edge costs
  are integers between 1 and W; an $\epsmore$-minimum $T$ join can be
  computed in $\apxO{m^{1.5} \log \reps}$ time.
\end{fact}
\begin{proof}[Proof sketch]
  Here we combine \reffact{t-join-gadget} with the min-cost perfect
  matching algorithm of \citet{gt-91} to obtain the desired running
  time.
\end{proof}

\begin{remark}
  The bottleneck of Christofides' algorithm is computing a $T$-join,
  where $T$ can have size $\sizeof{T} = \bigOmega{n}$. Moreover, it
  suffices to compute an $\epsmore$-approximate minimum $T$-join for
  $\eps = \frac{2}{n}$, since the MST has cost at most
  $\parof{1 - \frac{1}{n}}$-fraction of the minimum cost tour.
  Combining the $\epsmore$-approximation algorithm \reffact{gt-91}
  with alternatively the reductions of \reffact{edmonds-65} or
  \reffact{t-join-gadget}, Christofides' algorithm can be implemented
  implemented in
  \begin{math}
    \apxO{m n + n^{2.5}}
  \end{math}
  or
  \begin{math}
    \apxO{m^{1.5}}
  \end{math}
  time.
\end{remark}

\paragraph{The $T$-join polytope:}

\begin{definition}
  The \emph{dominant} of a polytope $P$ is the set
  $\setof{P + x \where x \geq \zeroes}$.
\end{definition}

\begin{fact}[{\citealp{ej-73}}]
  The dominant of the $T$-join polytope is the set of vectors
  $x \in \reals^{\edges}$ such that
  \begin{align*}
    x \geq \zeroes \text{ and }
    \sum_{e \in \delta(S)} x_e \geq 1 \text{ for each } S \subseteq V
    \text{ with }
    \sizeof{S \cap T} \text{ odd.}
    \labelthisequation[\textsf{JD}]{t-join-dominant}
  \end{align*}
\end{fact}

\begin{observation}[\cite{wolsey-80}]
  \labelobservation{2ecss-contains-t-joins} %
  Suppose $x \in \reals^n$ is feasible for the LP \refequation{2ecss},
  then $\frac{x}{2}$ is in the dominant of the
  $T$-join polytope for any even $T \subseteq V$.
\end{observation}

\section{Sparsifying solutions to \refequation{2ecss}}
\labelsection{sparsification} In this section we prove
\reftheorem{sparsifying-2ecss}.  The idea is straight forward in
retrospect. Let $x$ be a feasible solution to \refequation{2ecss}. We
can view $x$ as capacities on the edges w/r/t which $G$ is
$2$-edge-connected.  and apply cut sparsification techniques to obtain
another solution $x'$ where $x'$ is sparse.  Cut sparsification is a
standard technique with many applications. Here we also have an
objective function that needs to be preserved, and a black box
sparsification does not suffice. We observe that random sampling
based cut sparsification techniques \cite{bk-15,fhhp-11} are based on
importance sampling and can be adjusted to preserve the cost of the
objective function.  We first give the high-level details of the
scheme from \cite{bk-15} and then build upon it to derive our result.

\begin{setting}
  Let $\graph = (\vertices,\edges,\weight)$ be a weighted undirected
  graph.
\end{setting}

\begin{definition}
  $\graph$ is \emph{$k$-connected} if the value of each cut in $G$ is at
  least $k$.
\end{definition}
\begin{definition}
  A \emph{$k$-strong component} is a maximal $k$-connected
  vertex-induced subgraph of $\graph$.
\end{definition}
\begin{definition}
  The \emph{strong connectivity} or \emph{strength} of an edge $e$,
  denoted by $\strength{e}$, is the maximum value of $k$ such that a
  $k$-strong component contains (both endpoints of) $e$.
\end{definition}
\begin{fact}[{\citealp[Lemma 4.11]{bk-15}}]
  \begin{math}
    \sum_{e \in \edges} \frac{\weight{e}}{\strength{e}} \leq n - 1.
  \end{math}
\end{fact}

\begin{fact}[{\citealp[Compression Theorem 6.2]{bk-15}}]
  \labelfact{sparsification} Let $p: \edges \to [0,1]$ be a set of
  probabilities on the edges of $\graph$. Let
  $H = (\vertices,\edges',\weight')$ be a random weighted graph where
  for each edge $e \in \edges$, $\edges'$ independently samples
  \begin{math}
    e
  \end{math}
  with weight $\weight{e}' = \frac{\weight{e}}{p_e}$ with probability
  $p_e$, and $e \notin \edges$ with probability $1 - p_e$. For
  $\delta \geq \bigOmega{\log n}$, if
  \begin{math}
    p_e \geq \min{1, \frac{\delta}{\strength{e}}}
  \end{math}
  for all $e \in \edges$, then with probability
  $1 - \expof{-\bigOmega{\eps^2 \delta}}$, every cut in $H$ has value
  between $\epsless$ and $\epsmore$ times its value in $G$.
\end{fact}

\begin{fact}[{\citealp[Theorem 6.5]{bk-15}}]
  \labelfact{apx-strengths} %
  In $\bigO{m \log^3 n}$ time, one can compute values
  $\apxstrength{e} \geq 0$ for each $e \in \edges$ such that
  \begin{math}
    \apxstrength{e} \leq \strength{e}
  \end{math}
  for each $e \in \edges$ and
  \begin{math}
    \sum_{e \in \edges} \frac{\weight{e}}{\apxstrength{e}} = \bigO{n}.
  \end{math}
\end{fact}

\begin{lemma}
  \labellemma{sparsification} Given a feasible solution $x$ to
  \refequation{2ecss}, and a non-negative cost function $\cost :
  \edges \to \nnreals$, an $\eps > 0$,
  there is a randomized algorithm that runs in $\bigO{m
    \log^3 n}$ time and with probability at least $(1 - 1/n^2)$,
  outputs another feasible point $y$ for \refequation{2ecss} such that
  (i) $\rip{\cost}{y} = \epsmore \rip{\cost}{x}$, (ii) $\support{y} \subseteq
  \support{x}$, and (iii)
  $\sizeof{\support{y}} = \bigO{\frac{n \log n}{\eps^2}}$ .
\end{lemma}

\begin{proof}
  For each edge $e \in \edges$, let $\strength{e}$ be the strength of
  edge $e$ w/r/t the weighted graph $(\graph, x)$.  By
  \reffact{apx-strengths}, we can compute approximate strengths
  $\apxstrength_e \in [0,\strength{e}]$ such that
  \begin{math}
    \bigO{\sum_{e \in \edges} \frac{x_e}{\apxstrength{e}}} %
    = %
    \bigO{n} %
  \end{math}
  in $\bigO{m \log^2 n}$ time. Let $\delta = d \log n$ where
  $d$ is a sufficiently large constant. For each edge $e \in \edges$, let
  \begin{math}
    p_e = \min{1, \frac{\delta x_e}{\eps^2 \apxstrength{e}}},
  \end{math}
  and let
  \begin{math}
    q_e = %
    \min{1, %
      \frac{\delta \cost{e} x_e}{\eps^2 \sum_{e' \in \edges} \cost{e'}
        x_{e'}}                 %
    },
  \end{math}
  and let $r_e = \max{p_e, q_e}$. We have
  \begin{math}
    \sum_{e \in \edges} r_e %
    \leq %
    \sum_{e \in \edges} p_e %
    + %
    \sum_{e \in \edges} q_e %
    = %
    \bigO{\frac{n \delta}{\eps^2}} %
    + %
    \bigO{\frac{\delta}{\eps^2}} %
    \leq %
    \bigO{\frac{n \delta}{\eps^2}}.
  \end{math}
  Let $H = (\vertices', \edges', x')$ be the random weighted graph
  where each edge $e \in \edges$ is independently sampled with weight
  $x'_e = x_e / r_e$ with probability $r_e$. By
  \reffact{sparsification} and the assumption that
  $x \in \refequation{2ecss}$, with probability
  $1 - \expof{- \bigOmega{\delta}}$, we have
  \begin{align*}
    \sum_{e \in \delta(S)} x'_e \in \epspm \sum_{e \in \delta(S)} x_e
    \geq %
    \epsless 2
  \end{align*}
  for all $S \subset V$. By the
  multiplicative Chernoff inequality, we also have
  \begin{math}
    \probof{\sum_{e \in \edges} c_e x_e' \geq \epsmore \sum_{e \in
        \edges} c_e x_e} \leq \expof{- \bigOmega{\delta}},
  \end{math}
  and
  \begin{math}
    \probof{ %
      \sizeof{\edges'} %
      \geq %
      \epsmore \bigO{\frac{n \delta }{\eps^2}} %
    } %
    \leq %
    \expof{-\delta / \eps^2}.
  \end{math}
  By the union bound, we have
  $\sum_{e \in \delta(S)} x'_e \geq \epsless 2$ for all
  $S \subset V$,
  $\sizeof{\support{x'}} \leq \bigO{\frac{n \delta}{\eps^2}}$, and
  $\sum_{e \in \edges} c_e x_e' \leq \epsmore \sum_{e \in \edges} c_e
  x_e$ with probability of $\geq 1 - \expof{-\bigOmega{\delta}}$. Then
  $y = \epsmore x'$ is feasible for \refequation{2ecss},
  has $\bigO{\frac{n \delta}{\eps^2}}$ nonzeroes, and has cost
  \begin{math}
    \sum_{e \in \edges} \cost{e} y_e %
    \leq %
    \epsmore^2 \sum_{e \in \edges} \cost{e} x_e.
  \end{math}
  We obtain the desired statement by choosing $d$ sufficiently large
  to ensure that the success probability is at least $(1-1/n^2)$ and
  by choosing a smaller $\eps'$ in the above analysis so that
  $(1+\eps')^2 \le (1+\eps)$.
\end{proof}

\reftheorem{sparsifying-2ecss} is an easy consequence of
\reflemma{sparsification} applied to \reftheorem{apx-2ecss}.

\section{Approximation schemes for Christofides' algorithm}

\reffigure{apx-christofides} describes our randomized algorithm that
yields a $\parof{\frac{3}{2} + \eps}$-approximation for \metrictsp via
the LP solution followed by sparsification. It basically implements
Christofides's algorithm on the sparsified graph.

\begin{figure}[t]
  \centering

  \begin{minipage}{.7\paperwidth}
    \ttfamily \bigskip \underline{apx-Christofides($\defweightedgraph$,$\eps$)}
    \begin{enumerate}
    \item Compute the minimum spanning tree $S$ in $\apxO{m}$
      time. Let $T$ be the set of odd-degree vertices in $S$.
    \item By \reftheorem{apx-2ecss}, compute an
      $\epsmore$-approximation $x$ to \refequation{2ecss} in
      $\apxO{m / \eps^2}$ time with probability $1 -
      1/\poly{m}$.
    \item By \reftheorem{sparsifying-2ecss}, in time
      $\apxO{\frac{m}{\eps^2}}$ and with probability
      $1 - \frac{1}{\poly{m}}$, compute a feasible
      point $y \in \reals^E$ for \refequation{2ecss}
      such that
      \begin{itemize}
      \item
        \begin{math}
          \sum_e \cost{e} y_e %
          \leq %
          \epsmore \sum_e \cost{e} x_e
        \end{math}
      \item
        $\sizeof{\support{y}} = \apxO{n / \eps^2}$.
      \end{itemize}
      Note that $y/2$ lies in the dominant of the
      $T$-join polytope, \refequation{t-join-dominant}.
    \item Let $H = (V, \support{y}, \cost)$ be the subgraph of edges
      with nonzero values in $y$. By \reffact{sparse-t-join},
      compute the minimum cost $T$-join $J \subseteq H$ in $\apxO{n^{1.5}}$
      time.
    \item Compute, shortcut, and return an Euler tour on the
      multigraph $S \cup J$.
    \end{enumerate}
    \smallskip
  \end{minipage}

  \medskip

  \begin{minipage}{.7\paperwidth}
    \caption{Pseudocode for a randomized
      $\apxmore \frac{3}{2}$-approximation algorithm to Metric TSP
      w/r/t the shortest path metric of a weighted undirected graph
      (see
      \reftheorem{apx-christofides}).\labelfigure{apx-christofides}}
  \end{minipage}
\end{figure}

\begin{namedtheorem*}{\reftheorem{fast-christofides-intro}}
  \labeltheorem{apx-christofides}
  In
  \begin{math}
    \apxO{\frac{m}{\eps^2} + \frac{n^{1.5}}{\eps^3}}
  \end{math}
  time, with probability at least $1 - \frac{1}{\poly{m}}$,
  \refroutine{apx-Christofides}{\graph}{\epsilon} returns a tour of
  cost at most
  \begin{math}
    \epsmore \frac{3}{2} \opt \refequation{subtour-elimination}.
  \end{math}
\end{namedtheorem*}

\begin{proof}
  With probability $1 - 1/\poly{m}$, \refroutine{apx-Christofides}
  computes a minimum spanning tree $S$ and the minimum cost $T$-join
  on the odd degree vertices of $S$ within a subgraph $H$ of
  $G$. Since every vertex in the multigraph $S \cup T$ has even
  degree, it has an Eulerian tour, which can be shortcut can returned.

  The cost of the tour is at most
  \begin{math}
    \sum_{e \in S} \cost{e} + \sum_{e \in J} \cost{e}.
  \end{math}
  Since a $\frac{n-1}{n}$-fraction of any feasible solutions to the
  \refequation{subtour-elimination} lies in the spanning tree
  polytope, the cost of $S$ is at most
  \begin{math}
    \sum_{e \in S} \cost{e} \leq
    \opt\refequation{subtour-elimination}.
  \end{math}

  To bound the cost of the $T$-join $J$, we first observe that by
  \reffact{subtour-elimination=2ecss} and \reftheorem{apx-2ecss},
  $\sum_{e \in \edges}\cost{e} x_e \leq \epsmore
  \opt\refequation{subtour-elimination}$.  Applying
  \reftheorem{sparsifying-2ecss} to $x$, with probability $1 -
  \frac{1}{\poly{m}}$, we output a vector $y$ that also lies in
  \refequation{2ecss}, has cost $\sum_{e} \cost{e} y_e \leq
  \epsmore^2 \opt \refequation{subtour-elimination}$, and has support of size
  $\norm{y}_0 = \bigO{n \lnof{n} / \eps^2}$. Since $y/2$ lies in the
  dominant of the $T$-join polytope, the minimum cost $T$-join $J$ in
  the support of $y$ has cost at most $\sum_{e \in J} \cost{e} \leq
  \sum_e \cost{e} y_e/2$. Thus, the cost of a $(1+\eps)$-approximate
  $T$-join is $(1+O(\eps)) \opt\refequation{subtour-elimination}$.

  To bound the running time, consider the steps as enumerated in
  \reffigure{apx-christofides}. Step 1, computing the minimum spanning
  tree, takes $\apxO{m}$ time. Step 2 computes an approximate solution
  $x$ to \refequation{2ecss}, which by \reftheorem{apx-2ecss} can be computed
  in $\apxO{m / \eps^2}$ time. Step 3 sparsifies $x$ to produce a
  point $y$ with
  \begin{math}
    \sizeof{\support{y}} %
    = %
    \bigO{\frac{n \log n}{\eps^2}}
  \end{math}
  time. Step 4 applies a $(1+1/n)$-approximate $T$-join algorithm
  from \reffact{sparse-t-join} to the subgraph induced by the support of
  $y$, which contains $\apxO{n/\eps^2}$ edges, and takes
  $\apxO{\frac{n^{3/2}}{\eps^3}}$-time w/r/t the input graph
  $\graph$. The two bottlenecks are approximating \refequation{2ecss}
  and computing the $T$-join, and the total running time is
  \begin{math}
    \apxO{\frac{m}{\eps^2} + \frac{n^{3/2}}{\eps^3}}.
  \end{math}
\end{proof}

\FloatBarrier

\paragraph{Remarks:} \citet{GenovaW17} did an experimental
evaluatation of several variants of ``Best-of-many Christofides''
heuristics for \metrictsp. These heuristics are based on choosing a
\emph{random} spanning tree $T$ from an appropriate distribution that
is based on an LP solution $x$, and then using that tree in the
Christofides's heuristic. One of these heuristics is to decompose a
scaled version of the LP solution $x$ into a convex combination of
trees, and use the technique of swap-rounding \citep{CVZ10} to
generate a random spanning tree from the convex combination. The work
in \cite{cq-17-soda} describes a near linear time algorithm to
decompose $x$ into a $(1-\eps)$-approximate convex combination of
spanning trees. The particular structure of the decomposition, we
believe, should allow one to implement swap-rounding step also in near
linear time. This would lead to an implementation whose running time
is similar to the one we have in this paper for the basic
Christofides heuristic. We also believe that the ideas in
\cite{cq-17-focs}, and the ones here, will extend to develop faster
approximation algorithms for the $s$-$t$-path TSP problem that has
received substantial attention recently; we refer the reader to
\citep{AnKS15,vygen-12,SeboV16}.

\paragraph{Acknowledgments:} CC thanks a conversation with Sanjeev
Khanna, and a talk by David Shmoys on TSP that reminded him of Wolsey's
analysis, for inspiration. Both happened at the Simons Institute,
Berkeley during the workshop ``Discrete Optimization via Continuous
Relaxation'', September 2017. We thank Xilin Yu for several discussions
on metric matchings and related problems.

\bibliographystyle{plainnat}
\bibliography{christofides}

\end{document}